\documentclass[11pt,a4paper,english]{amsart}

\usepackage{color}

\usepackage{babel}
\usepackage[latin1]{inputenc}
\usepackage{amsmath}
\usepackage{amsthm}
\usepackage{amsfonts}
\usepackage{indentfirst}
\usepackage{graphicx}
\usepackage{amssymb}
\usepackage[mathscr]{eucal}
\usepackage{tikz}
\usepackage{caption}
\usepackage{amsthm}
\usepackage{amscd}

\usepackage[T1]{fontenc}

\newtheorem{theorem}{Theorem}

\newtheorem{corollary}{Corollary}

\newtheorem{proposition}{Proposition}

\newcommand{\card}{{\rm card}}

\newcommand{\trqp}{\mathrm{tr}_{q/p}}
\newcommand{\ps}{{\perp_s}}
\newcommand{\pt}{{\perp_{ts}}}
\newcommand{\ph}{{\perp_h}}

\title[EAQECCs over arbitrary finite fields]{Entanglement-assisted quantum error-correcting codes over arbitrary finite fields}
\author{C. Galindo, F. Hernando, R. Matsumoto and D. Ruano}
\curraddr{\texttt{Carlos Galindo and Fernando Hernando:} Instituto
Universitario de Matem\'aticas y Aplicaciones de Castell\'on and
Departamento de Matem\'aticas, Universitat Jaume I, Campus de Riu
Sec. 12071 Castell\'{o} (Spain)\\
\texttt{Ryutaroh Matsumoto:} Department of Information and Communications Engineering,
Tokyo Institute of Technology, 152-8550 Tokyo, Japan,
and Department of Mathematical Sciences, Aalborg University,
  Denmark\\
\texttt{Diego Ruano:} IMUVA-Mathematics Research Institute, Universidad de Valladolid, 47011 Valladolid (Spain).}
\email{
galindo@uji.es;
carrillf@uji.es;
ryutaroh@ict.e.titech.ac.jp;
diego.ruano@uva.es}
\date{}
\thanks{Supported by the Spanish Ministry of Economy/FEDER: grants MTM2015-65764-C3-1-P, MTM2015-65764-C3-2-P, MTM2015-69138-REDT and RYC-2016-20208 (AEI/FSE/UE), the University Jaume I: grant UJI-B2018-10, Spanish Junta de CyL: grant VA166G18, and JSPS Grant No.\ 17K06419.}

\begin{document}

\begin{abstract}
We prove that the known formulae for computing the optimal number of maximally entangled pairs required  for entanglement-assisted quantum error-correcting codes (EAQECCs) over the binary field hold for codes over arbitrary finite fields as well. We also give a Gilbert-Varshamov bound for EAQECCs and  constructions of EAQECCs coming from punctured self-orthogonal linear codes which are valid for any finite field. \end{abstract}

\maketitle

\section{Introduction}
The Shor's proposal of using quantum error correction for reducing decoherence in quantum computation \cite{shor95} and his polynomial-time algorithms for prime factorization and discrete logarithms on quantum computers \cite{22RBC} clearly illustrate the feasibility  and importance  of quantum computation and quantum error-correcting.

Most of the quantum error-correcting codes (QECCs) come from classical codes. The first known stabilizer quantum codes were binary  \cite{calderbank98,gott}. Later, stabilizer codes over any finite field were introduced and studied, they are of particular interest because their utility in fault-tolerant computation. Following \cite{ketkar06}, one can obtain QECCs of length $n$ over a finite field $\mathbb{F}_q$ from additive codes included in $\mathbb{F}_q^{2n}$ which are self-orthogonal with respect to a trace symplectic form. Working on this construction, QECCs of length $n$ over $\mathbb{F}_q$ can be derived from classical self-orthogonal codes with respect to the Hermitian inner product included in $\mathbb{F}_{q^2}^{n}$, and also from codes in $\mathbb{F}_{q}^{n}$ which are self-orthogonal with respect to the Euclidean inner product.

The previously mentioned self-orthogonality conditions (or some similar requirements of inclusion of codes in the dual of others) prevent the usage of many common classical codes for providing quantum codes. Brun, Dvetak and Hsieh in \cite{brun1} proposed to share entanglement between encoder and decoder to simplify the theory of quantum error-correction and increase the communication capacity. With this new formalism, entanglement-assisted quantum stabilizer codes can be constructed from any classical linear code giving rise to entanglement-assisted quantum error-correcting codes (EAQECCs). A formula to obtain the optimal number of ebits required for a binary entanglement-assisted code of  Calderbank-Shor-Steane (CSS) type was showed in \cite{hsieh} and formulae for more general constructions, including the consideration of duality with respect to symplectic forms were given in \cite{wilde08}. In fact, \cite{wilde08} proves that the optimal number $c$ of ebits required for a binary   entanglement-assisted quantum error-correcting code with generator matrix $(H_X|H_Z)$ is $\mathrm{rank}(H_X H_Z^T - H_Z H_X^T)/2$, where the superindex $T$ means transpose. Remark 1 in that paper states, without a proof, that the same formula holds when considering codes over finite fields $\mathbb{F}_p$, $p$ being a prime number, a proof can be found in \cite{schin}.

Recently, one can find in the literature some papers where the above formula (or formulae derived from it) are used for determining the entanglement corresponding to EAQECCs over arbitrary finite fields (see, for instance \cite{chen,Guenda18,Liu18,qian}). Although it holds for any finite field, we have found no proof in the literature and, thus, this work fills this gap. Therefore, this paper is devoted to prove formulae for the minimum required number $c$ of pairs of maximally entangled quantum states, corresponding to EAQECCs codes obtained from linear codes $C$   over any finite field, by using symplectic forms, or Hermitian or Euclidean inner products. We also show (see Subsection \ref{geometric}) that in the Hermitian and Euclidean cases,  $c$ is easy to compute when one chooses, as a basis of the linear code $C$ of length $n$,  a subset of those vectors giving rise to a geometric decomposition of the coordinate space of dimension $n$ that contains $C$ \cite{MetricStructure}.

In \cite{lai12}, a Gilbert-Varshamov type formula for the existence of binary EAQECCs was presented. Still with the idea of extending the binary case to the general one and with the help of our study of entanglement-assisted codes, we give a Gilbert-Varshamov type formula which is valid for any finite field. Furthermore, we will also provide conditions of existence and parameters of EAQECCs coming from classical self-orthogonal codes (say $C$) over any finite field. Since fewer qudits should be transmitted through a noisy channel, they perform better. Constructions of this type have been considered in the binary case for giving a coding scheme with imperfect e-bits \cite{lai12}.

Theorems \ref{thm:fp}, \ref{thm:hermitian} and \ref{thm:euclidean} contain our results about the entanglement required for EAQECCs over arbitrary finite fields. Section \ref{sect2} also explains how, in the Hermitian and Euclidean cases, nice bases of the vector spaces that contain the supporting linear codes allow us to get the corresponding required number $c$.  Section \ref{sect3} is devoted to state the mentioned Gilbert-Varshamov type bound and Section \ref{sec:wip} contains our results about EAQECCs coming from QECC  by considering symplectic, Hermitian or Euclidean duality.

\section{EAQECCs over $\mathbb{F}_q$}
\label{sect2}

The first three subsections of this section are devoted to prove formulae for computing the optimal entanglement corresponding to EAQECCs over arbitrary finite fields when considering symplectic forms, or Hermitian or Euclidean inner products.

\subsection{The symplectic case}
\label{symplectic}
Let $p$ be a prime number and $q$ a positive power $q=p^m$. Denote by $\mathbb{F}_q$ the finite field with $q$ elements. We also write $\mathbb{C}$ the field of complex numbers and $\mathbb{C}^r$, $r$ a positive integer, the $r$-coordinate space over $\mathbb{C}$.

Let $n$ be a positive integer, it is known (see for instance \cite[Theorem 13]{ketkar06}) that an $((n, K,d))_q$ stabilizer quantum code over $\mathbb{F}_q$ can be obtained from an additive code $C \subseteq \mathbb{F}_q^{2n}$ of size $q^n/K$ such that $C \subseteq C^{\perp_{ts}}$, and $\mathrm{swt} (C^{\perp_{ts}} \setminus C)= d$ when $K \geq 1$ and $d= \mathrm{swt} (C)$ otherwise. In the above result, we have considered the following notation which will be used in this paper as well. The symbol $\perp_{ts}$ means dual with respect to the {\it trace-symplectic form} on $\mathbb{F}_q^{2n}$:
\[
\left(\vec{a}|\vec{b}\right) \cdot_{ts} \left(\vec{a'}|\vec{b'}\right) = \trqp \left(\vec{a} \cdot \vec{b'} - \vec{a'} \cdot \vec{b}\right) \in \mathbb{F}_p,
\]
where $\left(\vec{a}|\vec{b}\right), \left(\vec{a'}|\vec{b'}\right) \in \mathbb{F}_q^{2n}$, $\vec{a} \cdot \vec{b'}$ and $\vec{a'} \cdot \vec{b}$ are Euclidean products, and $\trqp: \mathbb{F}_q \rightarrow \mathbb{F}_p$, $$ \trqp(x)= x + x^p + \cdots + x^{p^{m-1}},$$ is the standard trace map. Also the symplectic weight is defined as
\[
\mathrm{swt} \left(\vec{a}|\vec{b}\right) = \card \left\{i \;| \; (a_i,b_i) \neq (0,0), 1 \leq i \leq n \right\},
\]
where $\vec{a}=(a_1,a_2, \ldots, a_n)$ and $\vec{b}=(b_1,b_2, \ldots, b_n)$.

We will also use the symplectic form on $\mathbb{F}_q^{2n}$ defined as
\[
\left(\vec{a}|\vec{b}\right) \cdot_{s} \left(\vec{a'}|\vec{b'}\right) = \left(\vec{a} \cdot \vec{b'} - \vec{a'} \cdot \vec{b}\right) \in \mathbb{F}_q,
\]
and the corresponding dual space for an $\mathbb{F}_q$-linear code $C \subseteq \mathbb{F}_q^{2n}$ will be denoted by $C^{\perp_s}$.

For the first part of this paper, we fix a trace orthogonal basis of $\mathbb{F}_q$ over $\mathbb{F}_p$, $B=\{\gamma_1, \gamma_2, \ldots, \gamma_m\}$. Recall that $B$ is a basis of $\mathbb{F}_q$ as a $\mathbb{F}_p$-linear space satisfying that the matrix
\[
M= \left(\trqp(\gamma_i \gamma_j)\right)_{1\leq i \leq m;\;1\leq j \leq m}
\]
is an invertible and diagonal matrix of size $m$ with coefficients in $\mathbb{F}_p$. The existence of a basis as $B$ is proved in \cite{sero}. We choose a basis as $B$ by convenience, but our results also hold if one considers any other basis. Now  consider the  $\mathbb{F}_p$-linear map
\[
h: \mathbb{F}_p^m \rightarrow \mathbb{F}_q, \;\; h(x_1,x_2, \ldots,x_m)= \sum_{i=1}^m x_i \gamma_i:= x.
\]
The map $h$ is an isomorphism of $\mathbb{F}_p$-linear spaces and for $x \in \mathbb{F}_q$, $h^{-1}(x)$ gives the coordinates of $x$ in the basis $B$.

Denote by $\Omega$ the inverse matrix of $M$, $\Omega$ is a size $m$ diagonal invertible matrix with entries in $\mathbb{F}_p$. Let $\omega_1, \omega_2, \ldots, \omega_m$ be its diagonal and define the map:
\[
\phi: \mathbb{F}_p^{2m}=\mathbb{F}_p^m \times \mathbb{F}_p^m \longrightarrow \mathbb{F}_q^2,
\]
given by
\begin{eqnarray*}
\phi\left((x_1,x_2, \ldots,x_m)|(y_1,y_2, \ldots,y_m)\right) &=& \left(\sum_{i=1}^m x_i \gamma_i, \sum_{i=1}^m y_i \omega_i \gamma_i\right)\\
&=& \left( h(x_1,x_2, \ldots,x_m), h[(y_1,y_2, \ldots,y_m)\Omega]\right).
\end{eqnarray*}
Taking into account that $\omega_i \in \mathbb{F}_p$, $B'=\{\omega_i \gamma_i\}_{i=1}^m$ is also a trace orthogonal basis of $\mathbb{F}_q$ over $\mathbb{F}_p$ whose matrix
$\left(\trqp(\omega_i \gamma_i \omega_j \gamma_j)\right)_{1\leq i \leq m;\;1\leq j \leq m}$ is $\Omega$.

In sum, $\phi$ is an isomorphism of $\mathbb{F}_p$-linear spaces and for $(x,y) \in \mathbb{F}_q^2$,
\[
\phi^{-1} (x,y) = \left(\phi^{-1}_1 (x,y), \phi^{-1}_2 (x,y)\right) \in \mathbb{F}_p^{2m},
\]where $\phi^{-1}_{1}$ (respectively, $\phi^{-1}_{2}$) is the first (respectively, second) projection of
$\phi^{-1}$ over the first (respectively, second) component of the Cartesian product $\mathbb{F}_p^m \times \mathbb{F}_p^m$. One has that $\phi^{-1}(x,y)$ simply gives a pair whose first components are the coordinates of $x$ in the basis $B$ and the second ones are  those of $y$ in the basis $B'$.

The above map can be extended to products of $n$ copies giving rise to
the map
\[
\phi^E: \mathbb{F}_p^{2mn}=(\mathbb{F}_p^m)^n \times (\mathbb{F}_p^m)^n \longrightarrow \mathbb{F}_q^{n} \times \mathbb{F}_q^{n} =\mathbb{F}_q^{2n},
\]
defined by
\begin{multline*}
\phi^E \left[\left( (a_{11}, \ldots, a_{1m}), \ldots, (a_{n1}, \ldots, a_{nm})|(b_{11}, \ldots, b_{1m}), \ldots, (b_{n1}, \ldots, b_{nm}) \right) \right] =\\
\left( h(a_{11}, \ldots, a_{1m}), \ldots, h(a_{n1}, \ldots, a_{nm})
| h[(b_{11}, \ldots, b_{1m})\Omega], \ldots, h[(b_{n1}, \ldots, b_{nm})\Omega]
\right).
\end{multline*}

Notice that $\phi^E$ is again an isomorphism of $\mathbb{F}_p$-linear spaces and $$\left(\phi^E\right)^{-1} \left(\vec{a}|\vec{b}\right)= \left( \left(\phi^E\right)^{-1}_1 \left(\vec{a}|\vec{b}\right)|\left(\phi^E\right)^{-1}_2 \left(\vec{a}|\vec{b}\right)\right),$$where $\left(\phi^E\right)^{-1}_1$ (respectively, $\left(\phi^E\right)^{-1}_2$) is the first (respectively, second) projection of
$\left(\phi^E\right)^{-1}_1$ over the first (respectively, second) component of the Cartesian product $(\mathbb{F}_p^m)^n \times (\mathbb{F}_p^m)^ n$. One has that $\left(\phi^E\right)^{-1}\left(\vec{a}|\vec{b}\right)$ equals the vector of coordinates of the element $\left(\vec{a}|\vec{b}\right) \in \mathbb{F}_q^{2n}$ in the basis of $\mathbb{F}_q^{2n}$ over $\mathbb{F}_p$ given by $\oplus_{n\; \mbox{\footnotesize{times}}} B \bigoplus \oplus_{n\; \mbox{\footnotesize{times}}} B'$.

Keeping the above notation, it is easy to deduce the following result in \cite{ashikhmin00}.

\begin{proposition}\label{prop:convert}
The following statements hold:
\begin{description}
  \item[a)] Let $x, y \in \mathbb{F}_q$, then
  \[
  \trqp (xy) = \left(\phi^{-1}_1 (x,y)\right) \cdot \left(\phi^{-1}_2 (x,y)\right),
  \]
  where $\cdot$ denotes the Euclidean product in $\mathbb{F}_p^m$.
  \item[b)] Let $\left(\vec{a}|\vec{b}\right), \left(\vec{a'}|\vec{b'}\right) \in \mathbb{F}_q^{2n}$, then
$$
 \left(\vec{a}|\vec{b}\right) \cdot_{st} \left(\vec{a'}|\vec{b'}\right)  
 =
 \left[
 \left(\phi^E\right)^{-1}_1 \left(\vec{a}|\vec{b}\right) | \left(\phi^E\right)^{-1}_2 \left(\vec{a}|\vec{b}\right)
 \right] \cdot_s \left[
 \left(\phi^E\right)^{-1}_1 \left(\vec{a'}|\vec{b'}\right) | \left(\phi^E\right)^{-1}_2 \left(\vec{a'}|\vec{b'}\right)
 \right],
$$
where $\cdot_s$ denotes the symplectic form in $\mathbb{F}_p^{2mn}$.
\end{description}
\end{proposition}

Our purpose in this section is to determine the optimal required number of pairs of maximally entangled states of  the EAQECC over an arbitrary finite field $\mathbb{F}_q$ that can be constructed from an $\mathbb{F}_q$-linear code $C \subseteq \mathbb{F}_q^{2n}$ with dimension $n-k$. Assume that $(H_X|H_Z)$ is an $(n-k)\times 2n$ generator matrix of $C$. The case when $m=1$ (i.e., $q$ is prime) is known (see \cite{wilde08,schin}) and the corresponding result is the following:

\begin{theorem}\label{thm:fp}
Let $C \subseteq \mathbb{F}_p^{2n}$ be an $(n-k)$-dimensional $\mathbb{F}_p$-linear space and $H = (H_X|H_Z)$ an $(n-k) \times 2n$ matrix  whose row space is $C$. Let $C' \subseteq \mathbb{F}_p^{2(n+c)}$ be an $\mathbb{F}_p$-linear
 space such that the projection of $C'$ to the $1, 2, \ldots, n, n+c+1, n+c+2, \ldots, 2n+c$-th coordinates is equal to $C$ and $C' \subseteq (C')^\ps$, where $c$ is the minimum required number of maximally entangled quantum states in $\mathbb{C}^p \otimes \mathbb{C}^p$. Then, $$2c = \mathrm{rank}\left(H_X H_Z^T - H_Z H_X^T\right).$$
The encoding quantum circuit is constructed from $C'$, and it encodes $k+c$ logical qudits in $\mathbb{C}^p \otimes \cdots (k+c\; \mbox{times}) \cdots \otimes \mathbb{C}^p$ into $n$ physical qudits using $c$ maximally entangled pairs. The minimum distance  is $d:= d_s \left(C^\ps \setminus (C\cap C^\ps) \right)$,
  where
\[
d_s\left(C^\ps \setminus (C\cap C^\ps)\right) = \min\left\{ \mathrm{swt}\left(\vec{a}|\vec{b}\right) \mid \left(\vec{a}|\vec{b}\right) \in C^\ps \setminus (C\cap C^\ps)\right\}.
\]
In sum,  $C$ provides an $[[n,k+c,d;c]]_p$ EAQECC over the field $\mathbb{F}_p$.
\end{theorem}

Theorem \ref{thm:fp} states that the required number of maximally entangled quantum states is given by the rank of the matrix $H_X H_Z^T - H_Z H_X^T$. Our next result shows that, even in the case of codes over an arbitrary finite field  $\mathbb{F}_q$, the above number depends only on the code $C$ and its symplectic dual.

\begin{proposition}\label{propdim}
  Let $C \subseteq \mathbb{F}_q^{2n}$ be a linear code over $\mathbb{F}_q$ and $(H_X|H_Z)$ its $(n-k)\times 2n$ generator matrix. Then,
  \[
  \mathrm{rank}\left(H_X H_Z^T - H_Z H_X^T\right) =
  \dim_{\mathbb{F}_q} C - \dim_{\mathbb{F}_q} (C \cap C^\ps).
  \]
\end{proposition}
\begin{proof}
Consider the  $\mathbb{F}_q$-linear map $f : \mathbb{F}_q^{2n} \rightarrow \mathbb{F}_q^{n-k}$ defined by $f\left(\vec{a}|\vec{b}\right) = \vec{a} H_Z^T - \vec{b} H_X^T$. Set $\mathrm{row}(H_X|H_Z)$ the row space of
  the matrix $(H_X|H_Z)$. Then we have
\begin{eqnarray*}
\mathrm{rank}\left(H_X H_Z^T - H_Z H_X^T\right) &=& \dim_{\mathbb{F}_q} f\left(\mathrm{row}(H_X|H_Z)\right)\\
&=& \dim_{\mathbb{F}_q} C - \dim_{\mathbb{F}_q} C \cap \ker(f)\\ &=& \dim_{\mathbb{F}_q} C - \dim_{\mathbb{F}_q} (C \cap C^\ps),
\end{eqnarray*}
which concludes the proof.
\end{proof}

Next, with the help of the above proposition, we prove that Theorem \ref{thm:fp} can be extended to codes over any finite field $\mathbb{F}_q$.

\begin{theorem}\label{thm:fq}
Let $C \subseteq \mathbb{F}_q^{2n}$ be an $(n-k)$-dimensional $\mathbb{F}_q$-linear space  and $H = (H_X|H_Z)$ a matrix  whose row space is $C$. Let $C' \subseteq \mathbb{F}_q^{2(n+c)}$ be an $\mathbb{F}_q$-linear
 space such that its projection to the coordinates $1, 2, \ldots, n, n+c+1, n+c+2, \ldots, 2n+c$ equals $C$ and $C' \subseteq (C')^\ps$, where $c$ is the minimum required number of maximally entangled quantum states in $\mathbb{C}^q \otimes \mathbb{C}^q$. Then, $$2c = \mathrm{rank}\left(H_X H_Z^T - H_Z H_X^T\right)  = \dim_{\mathbb{F}_q} C - \dim_{\mathbb{F}_q} \left(C \cap C^\ps\right).$$
The encoding quantum circuit is constructed from $C'$, and it encodes $k+c$ logical qudits in $\mathbb{C}^q \otimes \cdots (k+c\; \mbox{times}) \cdots \otimes \mathbb{C}^q$ into $n$ physical qudits using $c$ maximally entangled pairs. The minimum distance is $d:= d_s \left(C^\ps \setminus (C\cap C^\ps) \right)$,
where $d_s$ is defined as in Theorem \ref{thm:fp}. In sum,  $C$ provides an $[[n,k+c,d;c]]_q$ EAQECC over the field $\mathbb{F}_q$.
\end{theorem}

\begin{proof}
One has that the inclusion $C^\ps \subseteq C^\pt$ holds since $\trqp(0)=0$. In addition,  $C^\pt \subseteq C^\ps$. Indeed, following \cite{ashikhmin00}, if $\left(\vec{a}|\vec{b}\right) \in C^\pt$, then $\left(\vec{a}|\vec{b}\right) \cdot_{ts} \left(\vec{x}|\vec{y}\right) =0$ for all $\left(\vec{x}|\vec{y}\right) \in C$. Taking into account that $\alpha \left(\vec{x}|\vec{y}\right) \in C$ for any $\alpha \in \mathbb{F}_q$, then $\trqp \left((\vec{a}|\vec{b}) \cdot_{s} \alpha (\vec{x}|\vec{y})\right) =0$ for all $\alpha$. This means that $\trqp\left(\alpha\left( (\vec{a}|\vec{b}) \cdot_{s} (\vec{x}|\vec{y})\right)\right)=0$ for all $\alpha$, which proves $\left(\vec{a}|\vec{b}\right) \cdot_{s} \left(\vec{x}|\vec{y}\right)=0$ and, therefore $\left(\vec{a}|\vec{b}\right) \in C^\ps$.

Now, using the same notation as at the beginning of this section, consider the code over the field $\mathbb{F}_p$, $C_0:= (\phi^E)^{-1} (C)$. It is clear that $\dim_{\mathbb{F}_p} (C_0) = m(n-k)$, and by Proposition \ref{prop:convert} and  the equality $C^\ps = C^\pt$, we have
\[
\dim_{\mathbb{F}_p} C_0 = m(n-k) = m \dim_{\mathbb{F}_q} C.
\]
Thus,
$$
     \dim_{\mathbb{F}_p} C_0 - \dim_{\mathbb{F}_p} C_0 \cap C_0^\ps\\
    = m \left( \dim_{\mathbb{F}_q} C - \dim_{\mathbb{F}_q} (C \cap C^\ps) \right).
  $$

This shows that, by Theorem \ref{thm:fp}, we have an entanglement-assisted quantum code encoding $m(k+c)$ qudits in $\mathbb{C}^p$ and consuming $mc$ maximally entangled states in $\mathbb{C}^p \otimes \mathbb{C}^p$. Using the map $\phi^E$ and the fact that $C^{\perp_s} = C^{\perp_{ts}}$, we have an entanglement-assisted quantum code encoding $(k+c)$ qudits in $\mathbb{C}^q$ and consuming $c$ maximally entangled states in $\mathbb{C}^q \otimes \mathbb{C}^q$. In fact, one can construct $C'_0 \subseteq \mathbb{F}_p^{2m(n+c)}$ in the same way as constructed $C'$ from $C$ in Theorem \ref{thm:fp}. Applying $\phi^E$ to the code $C'_0$, we get the code $C'$ in the statement with the claimed properties. The minimum distance follows from \cite[Section III]{ketkar06}.
\end{proof}

\subsection{The Hermitian case}
\label{Hermitian}
In this subsection, we specialize the results in Subsection \ref{symplectic} by considering the Hermitian inner product instead of a symplectic form. With the above notation, consider the finite field $\mathbb{F}_{q^2}$ and a normal basis $\{w, w^q\}$ of $\mathbb{F}_{q^2}$ over $\mathbb{F}_{q}$. Fix a positive integer $n$ and, following \cite{ketkar06}, define a trace-alternating form over $\mathbb{F}_{q^2}^n$ as
\[
\vec{x} \cdot_a \vec{y} = \trqp \left( \frac{\vec{x} \cdot \vec{y}^q - \vec{x}^q \cdot \vec{y}}{w^{2q}-w^2} \right),
\]
where $\vec{z}^q$, $\vec{z} \in \mathbb{F}_{q^2}^n$, means the componentwise $q$-power of $\vec{z}$.
The map $\varphi: \mathbb{F}_{q}^{2n} \rightarrow \mathbb{F}_{q^2}^{n}$ given by $\varphi \left(\vec{a}|\vec{b}\right) = w \vec{a} + w^q \vec{b}$ is bijective and isometric because the symplectic and the Hamming weights of $\left(\vec{a}|\vec{b}\right)$ and $\varphi\left(\vec{a}|\vec{b}\right)$ coincide. In addition, for $\left(\vec{a}|\vec{b}\right), \left(\vec{a'}|\vec{b'}\right) \in \mathbb{F}_{q}^{2n}$, it holds that
\[
\left(\vec{a}|\vec{b}\right) \cdot_{ts} \left(\vec{a'}|\vec{b'}\right) = \varphi\left(\vec{a}|\vec{b}\right) \cdot_{a} \varphi\left(\vec{a'}|\vec{b'}\right).
\]

Recall that the Hermitian inner product of two vectors $\vec{x}, \vec{y} \in \mathbb{F}_{q^2}^{n}$ is defined to be $\vec{x} \cdot_h \vec{y} = \vec{x}^q \cdot \vec{y}$, where $\cdot$ means Euclidean product, and that, in \cite{ketkar06}, it is proved that for a  $\mathbb{F}_{q^2}$-linear code $D$,  the dual codes with respect to the products $\cdot_a$ and $\cdot_h$ coincide. With the above ingredients, we are ready to prove the next proposition which will allow us to state and prove our theorem on EAQECCs over arbitrary finite fields by considering Hermitian inner product.

\begin{proposition}\label{propdimh}
  Let $C \subseteq \mathbb{F}_{q^2}^n$ be a code over $\mathbb{F}_{q^2}$ of dimension $(n-k)/2$ for some positive integer $k$. Let $H$ be its generator matrix. Then
\[
  \mathrm{rank}(HH^*) =
  \dim_{\mathbb{F}_{q^2}} C - \dim_{\mathbb{F}_{q^2}} (C \cap C^\ph),
\]
where $H^*$ is the $q$th power of the transpose matrix of $H$.
\end{proposition}
\begin{proof}
Define the $\mathbb{F}_{q^2}$-linear map $f : \mathbb{F}_{q^2}^{n} \rightarrow \mathbb{F}_{q^2}^{(n-k)/2}$, given by $f(\vec{a}) =  \vec{a} H^*$. Then,
\begin{eqnarray*}
\mathrm{rank}(H H^*) &=& \dim_{\mathbb{F}_{q^2}}  f(\mathrm{row}(H))\\
&=& \dim_{\mathbb{F}_{q^2}} C - \dim_{\mathbb{F}_{q^2}} (C \cap \ker(f))\\ &=& \dim_{\mathbb{F}_{q^2}} C - \dim_{\mathbb{F}_{q^2}} (C \cap C^\ph).
\end{eqnarray*}
\end{proof}

\begin{theorem}\label{thm:hermitian}
Let $C \subseteq \mathbb{F}_{q^2}^{n}$ be an $(n-k)/2$-dimensional code over $\mathbb{F}_{q^2}$, for suitable integers $n$ and $k$.  Denote by $H$ its generator matrix. Let $C' \subseteq \mathbb{F}_{q^2}^{(n+c)}$ be an $\mathbb{F}_{q^2}$-linear space whose projection to the coordinates $1, 2, \ldots, n$ equals $C$ and satisfies $C' \subseteq (C')^{\perp_h}$, where $c$ is the minimum required number of maximally entangled quantum states in $\mathbb{C}^q \otimes \mathbb{C}^q$. Then, $$c = \mathrm{rank} \left(H H^* \right)  = \dim_{\mathbb{F}_{q^2}} C - \dim_{\mathbb{F}_{q^2}} \left(C \cap C^{\perp_h}\right).$$
The encoding quantum circuit is constructed from $C'$, and it encodes $k+c$ logical qudits in $\mathbb{C}^q \otimes \cdots (k+c\; \mbox{times}) \cdots \otimes \mathbb{C}^q$ into $n$ physical qudits using $c$ maximally entangled pairs. The minimum distance is $d:= d_H \left(C^{\perp_h} \setminus (C\cap C^{\perp_h}) \right)$,
where $d_H$ is defined as the minimum Hamming weight of the vectors in the set $C^{\perp_h} \setminus \left(C\cap C^{\perp_h} \right)$. In sum,  $C$ provides an $[[n,k+c,d;c]]_q$ EAQECC over the field $\mathbb{F}_q$.
\end{theorem}

\begin{proof}
With the above notation, consider the code $C'$ in $\mathbb{F}_{q^2}^n$ of dimension $n-k$ whose generator matrix is
\[
\mathcal{H} = \left(
                \begin{array}{c}
                  \omega H \\
                  \omega^q H \\
                \end{array}
              \right)
\]
and set $C_0 = \varphi^{-1} (C')$ the corresponding code in  $\mathbb{F}_{q}^{2n}$. Since $\varphi$ is an isometry, to obtain the value $2c$ corresponding to $C_0$, it suffices to compute the rank of the matrix given by the form $\cdot_a$ which is $\mathcal{J}= \mathrm{tr}_{q^2/q} \left((H H^* - H^q H^T)/ \lambda \right)$, where $\lambda = \omega^{2q} - \omega^2$ and $\mathrm{tr}_{q^2/q}$ the trace map from $\mathbb{F}_{q^2}$ to $\mathbb{F}_{q}$. Now, setting
\[
\mathcal{Z}= \left(
               \begin{array}{cc}
                 \omega^{q+1} & \omega^2 \\
                 \omega^{2q} & \omega^{q+1} \\
               \end{array}
             \right),
\]
it holds that $\mathcal{J}= (2/\lambda) \left(ZH H^*- Z^T H^q H^T \right)$. Performing elementary operations, we get that $ \mathrm{rank} (\mathcal{J}) = 2 \: \mathrm{rank} \left( H H^*\right)$. Finally, by our previous considerations, $\dim_{\mathbb{F}_q} C_0 = n-k$, $\dim_{\mathbb{F}_q} (C_0 \cap C_0^\ps) = 2c$, and
\[d_H \left(C^\ph \setminus C^\ph \cap C\right) = d_s \left(C_0^\ps \setminus (C^\ps \cap C_0) \right),
 \]
which proves our statement by Theorem \ref{thm:fq}.
\end{proof}

The following corollary is an immediate consequence of the above result.

\begin{corollary}
Let $C$ be an $[n,k,d]_{q^2}$ linear code over $\mathbb{F}_{q^2}$ and set $H$ a parity check matrix of $C$. Then, there exist an $[[n,2k-n+c,d;c]]_q$ EAQECC where $c= \mathrm{rank}(H H^*)$, $H^*$ being the $q$th power of the transpose matrix $H^T$.
\end{corollary}

\subsection{The Euclidean case}
In this section we will show that EAQECCs over any finite field $\mathbb{F}_{q}$ can be obtained through a CSS construction, where the Euclidean inner product is considered, and carried out with two $\mathbb{F}_{q}$-linear codes $C_1$ and $C_2$ of length $n$.  Assume that $C_1$ (respectively, $C_2$) has dimension $k_1$ and generator matrix  $H_1$ (respectively, $k_2$ and $H_2$). Before stating our result, we give the following proposition which will be used in its proof.

\begin{proposition}\label{propdimE}
With the above notations, it holds that
\begin{equation}
\label{eq1}
 \mathrm{rank}(H_1H_2^T) =  \dim_{\mathbb{F}_{q}} C_1 - \dim_{\mathbb{F}_{q}} (C_1 \cap C_2^\perp),
\end{equation}
and
\begin{equation}
\label{eq2}
\mathrm{rank}(H_2H_1^T) =  \dim_{\mathbb{F}_{q}} C_2 - \dim_{\mathbb{F}_{q}} (C_2 \cap C_1^\perp),
\end{equation}
where $\perp$ means Euclidean dual.
\end{proposition}
\begin{proof}
To prove Equality (\ref{eq1}), consider the $\mathbb{F}_q$-linear map $f : \mathbb{F}_{q}^{n} \rightarrow \mathbb{F}_{q}^{k_2}$ defined by the matrix $H_2^T$, that is  $f(\vec{a}) = \vec{a}H_2^T$. Then
\begin{eqnarray*}
\mathrm{rank}(H_1 H_2^T) &=& \dim_{\mathbb{F}_{q}} f(\mathrm{row}(H_1))\\
&=& \dim_{\mathbb{F}_{q}} C_1 - \dim_{\mathbb{F}_{q}} (C_1 \cap \ker(f))\\ &=& \dim_{\mathbb{F}_{q}} C_1 - \dim_{\mathbb{F}_{q}} (C_1 \cap C_2^\perp).
\end{eqnarray*}
Equality (\ref{eq2}) follows analogously from the map given by $H_1^T$.
\end{proof}

Next we state the main result in this section.

\begin{theorem}\label{thm:euclidean}
Let $C_1$ and $C_2$ be two linear codes over $\mathbb{F}_{q}$ included in $\mathbb{F}_q^{n}$ with respective dimensions $k_1$ and $k_2$ and generator matrices $H_1$ and $H_2$. Then, the code $C_0 = C_1 \times C_2 \subseteq \mathbb{F}_q^{2n}$ gives rise to an EAQECC which encodes $n-k_1-k_2 + c$ logical qudits into $n$ physical qudits using the minimum required of maximally entangled pairs $c$, which is
\[
c = \mathrm{rank}(H_1H_2^T) =  \dim_{\mathbb{F}_{q}} C_1 - \dim_{\mathbb{F}_{q}} (C_1 \cap C_2^\perp).
\]
The minimum distance of the entanglement-assisted quantum code is larger than or equal to
\[
d:= \min \left\{ d_H\left(C_1^\perp \setminus (C_2 \cap C_1^\perp)\right), d_H\left(C_2^\perp \setminus  (C_1 \cap C_2^\perp)\right) \right\}.
\]
In sum, one gets an $[[n, n-k_1-k_2+c,d;c]]_q$ EAQECC.
\end{theorem}
\begin{proof}
It suffices to notice that $\dim_{\mathbb{F}_{q}} C_0 = k_1 + k_2$, $C_0^\ps = C_2^\perp \times C_1^\perp$, and
  \begin{eqnarray*}
    && \dim_{\mathbb{F}_{q}} C_0 - \dim_{\mathbb{F}_{q}} (C_0^\ps\cap C_0) \\
    &=& \dim_{\mathbb{F}_{q}} \left(C_1 \times C_2\right) -\dim_{\mathbb{F}_{q}} \left((C_2^\perp \cap C_1)  \times (C_1^\perp \cap C_2)\right)\\
    &=& \left( \dim_{\mathbb{F}_{q}} C_1 - \dim_{\mathbb{F}_{q}} (C_2^\perp \cap C_1) \right) + \left(\dim_{\mathbb{F}_{q}} C_2 - \dim_{\mathbb{F}_{q}} (C_1^\perp \cap C_2) \right)\\
    &=& 2c.
  \end{eqnarray*}
  By construction, we have that $$d_s(C_0 \setminus C_0^\ps)
  \geq \min \left\{ d_H\left(C_1^\perp \setminus (C_2 \cap C_1^\perp)\right),
  d_H\left(C_2^\perp \setminus (C_1 \cap C_2^\perp)\right)\right\},$$
and then our statement follows from Theorem \ref{thm:fq}.
\end{proof}

\subsection{Geometric decomposition of  the coordinate space}
\label{geometric}
In this subsection we consider only the Hermitian and Euclidean cases, and we will explain that the required number of maximally entangled pairs is easy to compute when the generators of the supporting $\mathbb{F}_q$-linear code $C$ in $\mathbb{F}_q^{n}$ are a subset of a basis of $\mathbb{F}_q^{n}$ with a special metric structure and which is said to be {\it compatible with a geometric decomposition}  of $\mathbb{F}_q^{n}$ (see \cite{MetricStructure}). Notice that, in the Hermitian case, $q$ should be $q^2$, however, for simplicity's sake and only in this subsection, we will use $q$ as a generic symbol which means a power of a prime in the Euclidean case or an even power of a prime in the Hermitian case. For avoiding to repeat notation, again only in this subsection,  $\langle \vec{a} , \vec{b} \rangle$ will mean either the Hermitian inner product $\vec{a}  \cdot_h \vec{b}$ or the Euclidean one $\vec{a}  \cdot \vec{b}$.

Let us introduce some notation, we say that $\{\vec{v_1}, \vec{v_2}\}$ are {\it geometric generators of a hyperbolic plane} if $\langle \vec{v_1} , \vec{v_1} \rangle = \langle \vec{v_2} , \vec{v_2} \rangle =0$ and $\langle \vec{v_1} , \vec{v_2} \rangle =1$. We say that $\{\vec{v_1}, \vec{v_2}\}$ are {\it geometric generators of an elliptic plane} if  $\langle \vec{v_1} , \vec{v_1} \rangle =0$ and  $\langle \vec{v_2} , \vec{v_2} \rangle = \langle \vec{v_1} , \vec{v_2} \rangle =1$. Finally, we say that $\vec{v}$ {\it generates a non-singular space} if $\langle \vec{v}, \vec{v} \rangle \neq 0$.

Let $C \subseteq \mathbb{F}_q^n$ and set $\left\{ \vec{v_1}, \vec{v_2}, \ldots , \vec{v_n} \right\}$ a basis of $\mathbb{F}_q^n$ such that $C$ is generated by  $\left\{\vec{v_i}\right\}_{i \in I}$ for $I \subseteq \{1, 2, \ldots , n\}$. We say that $C$ is {\it compatible with a geometric decomposition} of $\mathbb{F}_q^n$ if $$\mathbb{F}_q^n = H_1 \oplus \cdots \oplus H_r \oplus L_1 \oplus \cdots \oplus L_s,$$ where the linear spaces from $H_1$, generated by  $\left\{\vec{v_1}, \vec{v_2}\right\}$, to $H_r $, generated by $\left\{\vec{v}_{2r-1} , \vec{v_r} \right\}$, are hyperbolic planes,  being the $\vec{v_i}$ geometric generators, and from $L_1$, generated by $\vec{v}_{2r+1}$, to $L_s$, generated by $\vec{v}_{2r+s} = \vec{v_n}$, are non-singular spaces. Then, we say that the vectors $\vec{v_1}, \vec{v_2}, \ldots, \vec{v_r}$ (and the indexes $1, 2, \ldots r$) are asymmetric and the vectors $\vec{v}_{r+1}, \vec{v}_{r+2}, \ldots, \vec{v}_n$ (and the indexes $r+1, r+2, \ldots, n$) are symmetric. Moreover, we also say that $(1,2)$, $\ldots$, $(r-1,r)$ are symmetric pairs.

In \cite{MetricStructure}, for the Euclidean inner product, it was proved that for characteristic different from 2, we can always obtain a basis $\left\{\vec{v_1}, \vec{v_2}, \ldots , \vec{v_n}\right\}$ of $\mathbb{F}_q^n$ such that $$\mathbb{F}_q^n = H_1 \oplus \cdots \oplus H_r \oplus L_1 \oplus \cdots \oplus L_s,$$ with $s \le 4$. For characteristic equal to 2, we may have a decomposition as in the case with characteristic different from 2, excepting when the vector $(1, 1, \ldots ,1 )$ belongs to the {\it radical} (or {\it hull}) of $C$, $C\cap C^\bot$ . In that particular case, it was given in \cite{MetricStructure} the following decomposition $$\mathbb{F}_q^n = H_1 \oplus \cdots \oplus H_r \oplus L_1 \oplus \cdots \oplus L_s \oplus E,$$ with $s \le 2$ and where $E$ is an elliptic plane.

Let $M = (\langle \vec{v_i} , \vec{v_j} \rangle )_{1\le i,j \le n}$, one has that $M$ has the form

$$
M= \left(
\begin{array}{l@{\hspace{0.4cm}}l@{\hspace{0.4cm}}l@{\hspace{0.4cm}}l@{\hspace{0.4cm}}l@{\hspace{0.4cm}}l@{\hspace{0.4cm}}l@{\hspace{0.4cm}}l}

0 & 1 & \nonumber  & \nonumber & \nonumber  & \nonumber & \nonumber & \nonumber  \\

1 & 0 & \nonumber  & \nonumber & \nonumber  & \nonumber & \nonumber & \nonumber  \\

\nonumber  & \nonumber & \ddots  & \nonumber & \nonumber  & \nonumber & \nonumber   & \nonumber  \\

\nonumber  & \nonumber & \nonumber  & 0  & 1 &  \nonumber & \nonumber & \nonumber   \\

\nonumber  & \nonumber &  \nonumber & 1  & 0 & \nonumber  & \nonumber & \nonumber   \\

\nonumber  & \nonumber &  \nonumber & \nonumber  & \nonumber & g_1  & \nonumber  &  \nonumber \\

\nonumber  & \nonumber &  \nonumber & \nonumber  & \nonumber & \nonumber  & \ddots  &  \nonumber  \\

\nonumber  & \nonumber &  \nonumber & \nonumber  & \nonumber & \nonumber  & \nonumber & g_s  \\

\end{array} \right),
$$
where $g_1, \ldots, g_s$ are non-zero, except for the case when the characteristic is 2 and $(1, 1, \ldots, 1)$ belongs to the radical of $C$; then we have that $M$ is equal to

$$
M= \left(
\begin{array}{l@{\hspace{0.4cm}}l@{\hspace{0.4cm}}l@{\hspace{0.4cm}}l@{\hspace{0.4cm}}l@{\hspace{0.4cm}}l@{\hspace{0.4cm}}l@{\hspace{0.4cm}}l@{\hspace{0.4cm}}l}

0 & 1 & \nonumber  & \nonumber & \nonumber  & \nonumber & \nonumber & \nonumber & \nonumber  \\

1 & 0 & \nonumber  & \nonumber & \nonumber  & \nonumber & \nonumber & \nonumber & \nonumber  \\

\nonumber  & \nonumber & \ddots  & \nonumber & \nonumber  & \nonumber & \nonumber   & \nonumber & \nonumber    \\

\nonumber  & \nonumber & \nonumber  & 0  & 1 &  \nonumber & \nonumber & \nonumber  & \nonumber   \\

\nonumber  & \nonumber &  \nonumber & 1  & 0 & \nonumber  & \nonumber & \nonumber & \nonumber  \\

\nonumber  & \nonumber &  \nonumber & \nonumber  & \nonumber & g_1  & \nonumber  &  \nonumber & \nonumber     \\

  \nonumber &  \nonumber & \nonumber  & \nonumber & \nonumber  & \nonumber & g_s & \nonumber & \nonumber \\

  \nonumber &  \nonumber & \nonumber  & \nonumber & \nonumber  & \nonumber & \nonumber & 0 & 1  \\

 \nonumber &  \nonumber & \nonumber  & \nonumber & \nonumber  & \nonumber & \nonumber & 1 & 1   \\

\end{array} \right).
$$

Now, let $i \in \{ 1 , 2, \ldots , n \}$. We define $i'$ as
\begin{itemize}

\item $i+1$ if $\vec{v_i}$ is the first generator of a hyperbolic plane $H$,

\item $i-1$ if $\vec{v_i}$ is the second generator of a hyperbolic plane $H$,

\item $i$ if $\vec{v_i}$ generates a one-dimensional linear space $L$,

\item $i+1$ if $\vec{v_i}$ is the first generator of an elliptic plane $E$.

\end{itemize}
Notice that we do not define $i'$ when $\vec{v_i}$ is the second geometric generator of an elliptic plane, because in this case, $(1, 1, \ldots, 1)$ is not in the radical of $C$ \cite{MetricStructure}. For $I \subseteq \{1, 2, \ldots, n \}$ we set $I'=\{ i' : i \in I \}$ and $I^\bot = \{1, 2, \ldots, n \} \setminus I'$. In this way we can
compute the dual code $C^\bot$ of a linear code $C$ generated by  $\left\{\vec{v_i}\right\}_{i\in I}$ easily since it is generated by $\left\{\vec{v_i}\right\}_{i \in I^\bot}$. Moreover, it can also be used to construct QECCs using the CSS construction  since   $C \subseteq C^\bot$ if and only if $I \subseteq I^\bot$. These kinds of decomposition arise naturally (i.e., for the usual generators) in some families of evaluation codes as BCH codes, toric codes, $J$-affine variety codes, negacyclic codes, constacylic codes, etc. and the previous approach has been exploited for constructing stabilizer quantum codes,  EAQECCs and LCD codes (see  \cite{qc4,LCD,Guenda18,lag2,Liu18} for instance).

The above paragraphs allow us to give a practical procedure for computing  the value $c$ given in Theorem \ref{thm:hermitian}, for the Hermitian product, and in Theorem \ref{thm:euclidean}, for the Euclidean product ($C_1=C_2=C)$. Assume that $C$ is a code  generated by $\left\{\vec{v_i}\right\}_{i\in I}$ compatible with a geometric decomposition of the corresponding coordinate space and write $I = I_R \sqcup  I_L$  (i.e. $I = I_R \cup  I_L$ and $I_R \cap I_L = \emptyset$), where  the radical of $C$, $C \cap C^\bot$, $\bot$ meaning dual with respect to the inner product $\langle ~ , ~ \rangle$, is generated by $\left\{\vec{v_i}\right\}_{i\in I_R}$. The radical of $C$ can be easily computed in this case: Indeed, given $i \in I$, one has that $i \in I_R$ if it holds that: $\vec{v_i}$ is the first generator of a hyperbolic plane $H$ and $i+1 \not\in I$,  $\vec{v_i}$ is the second generator of a hyperbolic plane $H$ and $i-1 \not\in I$, or $\vec{v_i}$ is the first generator of an elliptic plane $E$. Otherwise, $i \in I_L$. An equivalent way to characterize $I_R$ is the following: $I_R$ consists of asymmetric indexes whose pair does not belong to $I$, and $I_L$ consists of symmetric indexes and pairs of asymmetric indexes that belong to $I$. Summarizing, one has that when one considers a suitable basis as above, then
\begin{eqnarray*}
c  &=&  \dim_{\mathbb{F}_q} C - \dim_{\mathbb{F}_q} \left(C \cap C^\bot\right)
 =  \mathrm{card} (I) - \mathrm{card} (I \cap I^\bot)\\
   &=&  \mathrm{card} (I) - \mathrm{card} (I_R) = \mathrm{card} (I_L).
\end{eqnarray*}

Note that we have an EAQECC with maximal entanglement when $I=I_L$, i.e., when $I_R = \emptyset$. This fact was used, for instance,  in \cite{Guenda18}.

\section{Gilbert-Varshamov-type sufficient condition of existence of entanglement-assisted codes}
\label{sect3}
In this section we give a Gilbert-Varshamov-type bound which is valid for EAQECCs over arbitrary finite fields. A similar bound was stated in \cite{lai14} for the binary case.
\begin{theorem}\label{thmgv}
Assume the existence of positive integers $n$, $k\leq n$, $\delta$, $c\leq (n-k)/2$
such that
  \begin{equation}
    \frac{q^{n+k}-q^{n-k-2c}}{q^{2n}-1}\sum_{i=1}^{\delta - 1}
         {n \choose i}(q^2-1)^i < 1. \label{eqgv}
  \end{equation}
Then there exists an $\mathbb{F}_q$-linear code $C \subseteq
  \mathbb{F}_q^{2n}$ such that $\dim_{\mathbb{F}_q} C = n-k$,
  $d_s(C^\ps \setminus (C^\ps\cap C)) \geq \delta$ and
  $\dim_{\mathbb{F}_q} C - \dim_{\mathbb{F}_q} (C^\ps\cap C)= 2c$.
\end{theorem}
\begin{proof}
We will use a close argument to the proof of the Gilbert-Varshamov bound for stabilizer codes \cite{calderbank97,ketkar06}. Let $\mathrm{Sp}(q,n)$ be the symplectic group over  $\mathbb{F}_q^{2n}$ \cite[Section 3]{grove02} and  $A(k,c)$ the set of $\mathbb{F}_q$-linear spaces $V \subseteq \mathbb{F}_q^{2n}$ such that $\dim_{\mathbb{F}_q} V = n-k$ and $$\dim_{\mathbb{F}_q} V - \dim_{\mathbb{F}_q} \left(V^\ps\cap V \right)= 2c.$$
For $\vec{0} \neq \vec{e} \in \mathbb{F}_q^{2n}$, define
\[
B(k,c,\vec{e})
= \left\{ V\in A(k,c) \mid \vec{e} \in V^\ps \setminus (V^\ps\cap V) \right\}.
\]
 Taking into account that the symplectic group acts transitively on $\mathbb{F}_q^{2n} \setminus \{\vec{0}\}$ \cite{aschbacher00,grove02}, it holds that for nonzero $\vec{e}_1$, $\vec{e}_2 \in \mathbb{F}_q^{2n}$, there exists $M \in \mathrm{Sp}(q,n)$ such that $\vec{e}_1 M= \vec{e}_2$, and, for $V_1, V_2 \in A(k,c)$, there exists $M \in \mathrm{Sp}(q,n)$ such that $V_1 M= V_2$.

Therefore for nonzero elements $\vec{e}_1, \vec{e}_2 \in \mathbb{F}_q^{2n}$ with $\vec{e}_1 M_1 = \vec{e}_2 \; \left(M_1 \in \mathrm{Sp}(q,n)\right)$ and some fixed linear space $V_1 \in A(k,c)$, we have the following chain of equalities:
\begin{eqnarray*}
  &&  \mathrm{card} \left(B(k, c, \vec{e}_1)\right) \\
  &=& \mathrm{card}\left(\{ V \in A(k,c) \mid \vec{e}_1 \in V^\ps \setminus (V^\ps\cap V) \}\right) \\
  &=& \mathrm{card}\left(\{ V_1 M  \mid \vec{e}_1 \in V^\ps M \setminus (V^\ps M \cap V M), M  \in \mathrm{Sp}(q,n) \}\right)\\
  &=& \mathrm{card} \left(\{V_1 M  M_1^{-1}  \mid \vec{e}_1 \in V^\ps M  M_1^{-1} \setminus (V^\ps M  M_1^{-1}\cap V M  M_1^{-1}), M  \in \mathrm{Sp}(q,n) \}\right)\\
  &=& \mathrm{card}\left(\{ V_1 M \mid \vec{e}_1 M_1 \in V^\ps M \setminus (V^\ps M \cap V M), M  \in \mathrm{Sp}(q,n) \} \right)\\
  &=& \mathrm{card}\left(\{ V_1 M  \mid \vec{e}_2 \in V^\ps M \setminus (V^\ps M\cap V M), M  \in \mathrm{Sp}(q,n) \} \right)\\
  &=& \mathrm{card}\left(\{ V \in A(k,c) \mid \vec{e}_2 \in V^\ps \setminus (V^\ps\cap V) \} \right) \\
  &=& \mathrm{card}\left(B(k, c, \vec{e}_2)\right).
\end{eqnarray*}

For each $V \in A(k,c)$, the number of vectors $\vec{e}$ in $\mathbb{F}_q^{2n}$ such that
$\vec{e} \in V^\ps \setminus (V^\ps\cap V)$ is $$\mathrm{card}(V^\ps)-\mathrm{card}\left(V^\ps\cap V\right)=q^{n+k}-q^{n-k-2c}.$$
The number of pairs $(\vec{e}$, $V)$ such that $\vec{0} \neq
\vec{e} \in V^\ps \setminus (V^\ps\cap V)$ is
\[
  \sum_{\vec{0} \neq \vec{e}\in \mathbb{F}_q^{2n}} \mathrm{card}\left(B(k,c,\vec{e})\right) = \mathrm{card}\left(A(k,c)\right) \left(q^{n+k}-q^{n-k-2c}\right),
  \]
  which implies
  \begin{equation}
    \frac{\mathrm{card}\left(B(k,c,\vec{e})\right)}{\mathrm{card}\left(A(k,c)\right)} = \frac{q^{n+k}-q^{n-k-2c}}{q^{2n}-1}.
    \label{eq101}
  \end{equation}
If there exists $V \in A(k,c)$
such that $V \notin B(k,c,\vec{e})$
for all $1 \leq \mathrm{swt}(\vec{e}) \leq \delta-1$,
then there will exist $V$ with the desired properties.
The number of vectors $\vec{e}$ such that
$1 \leq \mathrm{swt}(\vec{e}) \leq \delta-1$
is given by
\begin{equation}
  \sum_{i=1}^{\delta - 1} {n \choose i}(q^2-1)^i. \label{eq103}
\end{equation}
By combining Equalities (\ref{eq101}) and (\ref{eq103}),  we see that Inequality (\ref{eqgv}) is a sufficient condition
for ensuring the existence of a code $C$ as in our statement. This ends the proof.
\end{proof}

To finish this section, we derive an asymptotic form of Theorem \ref{thmgv}.
\begin{theorem}\label{thmasymp}
Let $R$, $\epsilon$ and $\lambda$  be nonnegative real numbers such that $R \leq 1$, $\epsilon < 1/2$, and $\lambda \leq (1-R)/2$. Let $h(x) :=  -x \log_q x -(1- x) \log_q (x-1)$  be the $q$-ary entropy function. For $n$ sufficiently large, the inequality
  \begin{equation}
    h(\epsilon) + \epsilon \log_q (q^2-1)  <  1-R, \label{eq104}\\
  \end{equation}
implies the existence of a code
  $C  \subseteq
  \mathbb{F}_q^{2n}$ over $\mathbb{F}_q$ such that
   $$\dim_{\mathbb{F}_q} C = \lceil n(1-R) \rceil, \;\;
  d_s(C^\ps \setminus (C^\ps\cap C)) \geq \lfloor n\epsilon\rfloor$$ and
  $$\dim_{\mathbb{F}_q} C - \dim_{\mathbb{F}_q} (C^\ps\cap C) = \lfloor 2n\lambda\rfloor.$$
\end{theorem}
\begin{proof}
It follows from Theorem \ref{thmgv} using a similar reasoning to that in \cite[Section III.C]{matsumotouematsu01}. 
\end{proof}

\section{EAQECCs coming from punctured QECCs}\label{sec:wip}

Our  final section gives parameters of EAQECCs obtained from punctured codes coming from self-orthogonal codes with respect to symplectic forms, or Hermitian or Euclidean inner products. Since fewer qudits should be transmitted through a noisy channel, they perform better. Let us start with the symplectic case.

\subsection{Symplectic form}\label{sec:wips}
Let $C \subseteq \mathbb{F}_q^{2n}$ be an $\mathbb{F}_q$-linear code. The {\it puncturing} of $C$ to the coordinate set $\{1$, \ldots,
$n-c\}$ is defined as the code of length $2(n-c)$ given by
\begin{multline*}
P(C) = \big\{ (a_1, \ldots, a_{n-c}| b_1, \ldots, b_{n-c}) \mid
(a_1, \ldots, a_n| b_1, \ldots, b_n) \in C \\
\mbox{for some} \;\; a_{n-c+1}, \ldots, a_n, b_{n-c+1}, \ldots, b_n\in \mathbb{F}_q \big\}.
\end{multline*}

In addition, the {\it shortening} of $C$ to the coordinate set $\{1$, \ldots,
$n-c\}$ is defined as the code
$$
S(C) = \{ (a_1, \ldots, a_{n-c}| b_1, \ldots, b_{n-c}) \mid \\
(a_1, \ldots, a_{n-c}, 0, \ldots, 0| b_1, \ldots, b_{n-c},
0, \ldots, 0)\in C \}.
$$

When we have a stabilizer code given by an $\mathbb{F}_q$-linear code $C$ such that $C \subseteq C^\ps \subseteq
\mathbb{F}_q^{2n}$,
we can construct an entanglement-assisted code
from $P(C) \subseteq \mathbb{F}_q^{2(n-c)}$.
By \cite{pless98},
$P(C)^\ps = S(C^\ps)$
and we deduce $$P(C) \cap P(C)^\ps = P(C)\cap S(C^\ps)
= S(C \cap C^\ps) = S(C).$$

The minimum distance of the constructed entanglement-assisted code is
$d_s(S(C^\ps) \setminus S(C))$ which is larger than or equal to $d_s(C^\ps \setminus C)$. Following \cite{pless98}, one can prove the following result.

\begin{proposition}
Assume that a positive integer $c$ satisfies $2c \leq d_H(C^\ps \setminus \{\vec{0}\})-1$,
  then
  $$\begin{array}{l}
    \dim_{\mathbb{F}_q} P(C) = \dim_{\mathbb{F}_q} C, \;\;\mbox{and}\\ \\
    \dim_{\mathbb{F}_q} P(C)\cap P(C)^\ps = \dim_{\mathbb{F}_q} S(C) = \dim_{\mathbb{F}_q} C - 2c.
  \end{array}$$
\end{proposition}

Summarizing these observations, we have the following theorem. Notice that a close result has been given in  \cite{lai12}  for binary codes.

\begin{theorem}\label{thm:newcode}
Let $C \subseteq \mathbb{F}_q^{2n}$ be an $\mathbb{F}_q$-linear code with $\dim_{\mathbb{F}_q} C = n-k$ and $C \subseteq C^\ps$. Assume that a positive integer $c$ satisfies $2c \leq d_H(C^\ps \setminus \{\vec{0}\})-1$, then the punctured code $P(C)$ provides an $$[[n-c, k, \geq d_s(C^\ps\setminus C);  c]]_q$$ entanglement-assisted code.
\end{theorem}

Our next two sections are devoted to give similar results but considering Hermitian or Euclidean inner product.

\subsection{Hermitian inner product}
Let $C \subseteq \mathbb{F}_{q^2}^n$  an $\mathbb{F}_q$-linear code. The {\it $h$-puncturing} of $C$ to the coordinate set $\{1, 2, \ldots,
n-c\}$ is the code of length $n-c$ defined as
$$
P_h(C) = \\ \left\{ (a_1, a_2, \ldots, a_{n-c}) \mid
(a_1, a_2, \ldots, a_n) \in C \;
\mbox{for some} \; a_{n-c+1}, \ldots, a_n \in \mathbb{F}_{q^2} \right\}.
$$
The {\it $h$-shortening} of $C$ to the coordinate set $\{1, 2, \ldots, n-c\}$ is  the code of length $n-c$ defined as
$$S_h(C) = \left\{ (a_1, a_2, \ldots, a_{n-c}) \mid (a_1, a_2, \ldots, a_{n-c}, 0, \ldots, 0) \in C \right\}.$$

The above concepts allows us to state the following theorem.

\begin{theorem}\label{thm:newcodeh}
Let $C \subseteq \mathbb{F}_{q^2}^n$ be an $\mathbb{F}_{q^2}$-linear code with $\dim_{\mathbb{F}_{q^2}} C = (n-k)/2$ and suppose that $C \subseteq C^\ph$. Let $c$ be a positive integer such that $c \leq d_H(C^\ph \setminus \{\vec{0}\})-1$, then the punctured code $P_h(C)$ provides an $$[[n-c, k, \geq d_H(C^\ph\setminus C);  c]]_q$$ entanglement-assisted code.
\end{theorem}

\begin{proof}
By the assumption, $\dim_{\mathbb{F}_{q^2}} P_h(C) = \dim_{\mathbb{F}_{q^2}} C$. By  a similar argument to that used in Subsection \ref{sec:wips}, we also see that $P_h(C) \cap P_h(C)^\ph = S_h(C)$. Now we have that $c \le d_H (C^\ph) -1$, so $\dim_{\mathbb{F}_{q^2}} P_h(C) = \dim_{\mathbb{F}_{q^2}} C$ and $\dim_{\mathbb{F}_{q^2}} S_h(C) = \dim_{\mathbb{F}_{q^2}} C - c$ \cite{pless98}. It also holds that $d_H(P_h(C)^\ph \setminus S_h(C))
  \geq d_H(C^\ph \setminus C)$ and this concludes the proof by Theorem \ref{thm:hermitian}.
\end{proof}

\subsection{Euclidean inner product}
Our result concerning Euclidean duality is the following:
\begin{theorem}\label{thm:newcodee}
Let $C_2 \subseteq C_1 \subseteq \mathbb{F}_{q}^n$ be two $\mathbb{F}_{q}$-linear codes such that $\dim C_i = k_i$, $1 \leq i \leq 2$. The standard construction of CSS codes uses $C_2 \times C_1^\perp$ as the stabilizer.
Assume that $c$ is a positive integer such that $$c \leq \min\left\{d_H(C_2^\bot \setminus \{\vec{0}\}),
  d_H(C_1 \setminus \{\vec{0}\})\right\}-1,$$
  then the punctured code $P_h(C_2)\times P_h(C_1^\perp)$ provides an
$$[[n-c, k_1-k_2, \geq \min \left\{ d_H(C_1\setminus C_2), d_H(C_2^\perp\setminus C_1^\perp) \right\}; c]]_q$$ entanglement-assisted code.
\end{theorem}
\begin{proof}
The assumption $c \leq \min\left\{d_H(C_2^\bot \setminus \{\vec{0}\}),
  d_H(C_1 \setminus \{\vec{0}\})\right\}-1$
 implies the following two equalities: $\dim_{\mathbb{F}_{q}} P_h(C_2) = \dim_{\mathbb{F}_{q}} C_2$ and $\dim_{\mathbb{F}_{q}} P_h(C_1^\perp) = \dim_{\mathbb{F}_{q}} C_1^\perp$. Therefore
 \[
 \dim_{\mathbb{F}_{q}} P(C_2 \times C_1^\perp) = \dim_{\mathbb{F}_{q}}  P_h(C_2) + \dim_{\mathbb{F}_{q}} P_h(C_1^\perp) = n-(k_1-k_2).
 \]
Furthermore it holds that
  \begin{multline*}
\dim_{\mathbb{F}_{q}} P (C_2 \times C_1^\perp ) \cap P(C_2 \times C_1^\perp)^\ps
=  \dim_{\mathbb{F}_{q}} S(C_2 \times C_1^\perp)\\
    = \dim_{\mathbb{F}_{q}} [S_h(C_2) \times S_h(C_1^\perp)]
    = \dim_{\mathbb{F}_{q}} S_h(C_2) + \dim_{\mathbb{F}_{q}} S_h(C_1^\perp )\\ = \left(\dim_{\mathbb{F}_{q}} S_h(C_2)-c \right) + \left(\dim_{\mathbb{F}_{q}} S_h(C_1^\perp)-c \right)
   = n-(k_1-k_2)-2c.
  \end{multline*}
Applying Theorem \ref{thm:newcode} to the code $C_2 \times C_1^\perp$, the proof is completed.
\end{proof}

\vspace{0.5cm}

\textbf{Acknowledgments:} We thank Francisco R. Fernandes and Ruud Pellikaan for pointing out a mistake in Theorem \ref{thmasymp} on an earlier version of this article. We thank Markus Grassl and Hualu Liu for pointing out a mistake in Section \ref{sec:wip} on an earlier version of this article.

\end{document}